\documentclass[twocolumn,pra,letterpaper]{revtex4}
\usepackage{hyperref}
\usepackage{latexsym,amsmath,amssymb,amsfonts,graphicx,color,amsthm}
\usepackage{enumerate}

\newcommand{\cA}{\mathcal{A}}
\newcommand{\cB}{\mathcal{B}}

\newcommand{\cE}{\mathcal{E}}

\newcommand{\cH}{\mathcal{H}}

\newcommand{\cL}{\mathcal{L}}
\newcommand{\cM}{\mathcal{M}}

\newcommand{\cQ}{\mathcal{Q}}

\newcommand{\cZ}{\mathcal{Z}}

\newcommand{\Id}{\mathbb{I}}
\newcommand{\tr}{\text{tr}}

\newtheorem{theorem}{Theorem}
\newtheorem{lemma}[theorem]{Lemma}
\newtheorem{corollary}[theorem]{Corollary}

\begin{document}

\title{Measures of disturbance and incompatibility for quantum measurements}
\author{Prabha Mandayam}
\affiliation{The Institute of Mathematical Sciences, Taramani, Chennai   - 600113, India. }
\author{M. D. Srinivas}
\affiliation{Centre for Policy Studies, Mylapore, Chennai - 600004, India.}
\date{\today}

\begin{abstract}
We propose a class of incompatibility measures for quantum observables based on quantifying the effect of a measurement of one observable on the statistics of the outcomes of another. Specifically, for a pair of observables $A$ and $B$ with purely discrete spectra, we compare the following two probability distributions: one resulting from a measurement of $A$ followed by a measurement of $B$ on a given state, and the other obtained from a measurement of $B$ alone on the same state. We show that maximizing the distance between these two distributions over all states yields a valid measure of the incompatibility of observables $A$ and $B$, which is zero if and only if they commute and is strictly greater than zero (and less than or equal to one) otherwise.

For finite dimensional systems, we obtain a tight upper bound on the incompatibility of any pair of observables and show that the bound is attained when the observables are totally non-degenerate and associated with mutually unbiased bases. In the process, we also establish an important connection between the incompatibility of a pair of observables and the maximal disturbances due to their measurements. Finally, we indicate how these measures of incompatibility and disturbance can be extended to the more general class of non-projective measurements. In particular, we obtain a non-trivial upper bound on the incompatibility of one L\"uders instrument with another.
\end{abstract}

\maketitle

A central feature of quantum theory that lies at the heart of several quantum information processing and cryptographic tasks is the existence of incompatible observables. In quantum theory, compatible observables correspond to a set of commuting self-adjoint operators. Since their eigen-projectors also commute, there exists a joint probability distribution associated with such a set of observables. For non-commuting observables, however, there does not exist a joint probability distribution (which is affine in the density operator) whose marginals give the distributions for the individual observables. Quantifying this incompatibility of a set of non-commuting observables is an intriguing question with consequences for both quantum foundations and quantum information theory.

Heisenberg's celebrated uncertainty relation provided the first quantitative statement on the incompatibility of a pair of conjugate observables~\cite{heisenberg}, by providing a bound on the product of their variances. More recently, state-independent bounds on the sum of uncertainties of general sets of observables have been obtained via entropic uncertainty relations (EURs)~\cite{deutsch, Maassen-Uffink, Wehner-Winter}, and these are often considered to be providing a measure of incompatibility. However, EURs give rise to trivial bounds when the observables share even a single common eigenvector. While the lower bound on the sum of uncertainties does indeed capture the incompatibility of a large class of observables, this lower bound cannot be thought of as a {\it measure} of incompatibility in general, since it vanishes when the set of non-commuting observables share even a single common eigenvector. Rather, as pointed out recently~\cite{incompatibility_BM},
uncertainty relations must be viewed as merely a {\it consequence} of the incompatibility of observables.

A new, operational approach to quantifying incompatibility was proposed in~\cite{incompatibility_BM}, based on the notion of {\it accessible fidelity}~\cite{Fuchs-Sasaki}. The measure $\cQ$ defined in~\cite{incompatibility_BM} captures the incompatibility of a set of (totally non-degenerate) non-commuting observables as manifest in the non-orthogonality of their eigenstates.

In this paper we consider a different operational setting for defining incompatibility, which is closer in spirit to the original formulation due to Heisenberg and others of the uncertainty relation. We introduce a class of incompatibility measures which are based on estimating the change due to a measurement of one observable on the statistics of the outcomes of another. If a pair of observables $A$ and $B$ commute, then, a measurement of $B$, which follows a measurement of $A$, yields the same measurement statistics as a measurement of $B$ alone, on all states. However, if $A$ and $B$ do not commute, $A$ and $B$ are not jointly measurable, and, there exist states on which a measurement of $A$ {\it disturbs} the system in such a way that a subsequent measurement of $B$ yields probabilities very different from those associated with a measurement of $B$ alone. The {\it distance} between these two probability distributions -- one resulting from a measurement of $B$ following a measurement of $A$ and the
other resulting from a measurement of $B$ alone -- can thus be viewed as a measure of the effect of a measurement of $A$ on the outcomes of a measurement of $B$, for each given state. Maximizing this over all the states of the system gives a measure of incompatibility that is naturally state-independent.

We are thus lead to an entire class of measures of incompatibility, which are obtained by choosing different measures of distance between probability distributions. Each of these incompatibility measures is zero for a pair of observables if and only if the observables commute. Furthermore, this class of measures always yields a strictly positive value for the incompatibility even when the observables in question share some common eigenstates but do not commute over the entire space -- unlike uncertainty relations, which invariably yield a zero lower bound in such cases.

The paper is organized as follows. In Section~\ref{sec:dist_measure} we review the relevant distance measures between probability distributions and use them to define a class of incompatibility measures between observables. In Section~\ref{sec:incompat_disturbance} we prove that the incompatibility of observable $A$ with $B$ is bounded by the maximal disturbance due to a measurement of $A$. In Section~\ref{sec:fidelity_bounds} we focus on the finite dimensional case and prove a tight upper bound for the fidelity-based incompatibility measure. We also show that the upper bound is attained for mutually unbiased observables. Finally, in Section~\ref{sec:incompat_gen}, we show how the measures introduced here can be extended for the case of general quantum measurements, beyond the class of projective measurements.

\section{A new class of incompatibility measures}\label{sec:dist_measure}
\subsection{Distance measures for classical probability distributions and quantum states}
Given a pair of discrete probability distributions $P \sim \{p_{i}\}$ and $Q \sim \{q_{j}\}$, we consider the following three measures of distance between $P$ and $Q$~\cite{Cha, Mathai_Rathie}:
\begin{itemize}
\item[(i)] Variational Distance or $L_{1}$-Distance :
\begin{equation}
 D_{1}(P,Q) \equiv \frac{1}{2}\sum_{i}\vert p_{i} - q_{i}\vert \nonumber 
\end{equation}
\item[(ii)] Fidelity-based Distance :
\begin{equation}
  D_{F}(P,Q) \equiv 1 - (F(P,Q))^{2}, \nonumber
\end{equation}
where the Fidelity $F(P,Q)$ (also known as the Bhattacharyya Distance) is defined as
\begin{equation}
 F(P,Q) \equiv \sum_{i}\sqrt{p_{i}}\sqrt{q_{i}} \nonumber 
\end{equation}
\item[(iii)] Chebyshev Distance or $L_{\infty}$-Distance :
\begin{equation}
 D_{\infty}(P,Q) \equiv \max_{i}\vert p_{i} -  q_{i}\vert \nonumber 
\end{equation}
\end{itemize}

All three distance measures satisfy:
\begin{equation}
0\leq D_{\alpha}(P,Q) \leq 1, \; (\alpha \in \{1,F,\infty\})\nonumber
\end{equation}
with $D_{\alpha}(P,Q) = 0$ if and only if $P$ and $Q$ are identical. Furthermore, $D_{1}(P,Q)$ and $D_{\infty}(P,Q)$ are metrics on the space of probability distributions, that is, they are symmetric \[ D_{\alpha}(P,Q) = D_{\alpha}(Q,P), \; (\alpha=1, \infty)\]
and, for three probability distributions $P, Q, S$, the triangle inequality holds:
\[ D_{\alpha}(P,Q) \leq D_{\alpha}(P,S) + D_{\alpha}(S,Q), \; (\alpha=1, \infty)  .  \]
The fidelity-based measure stands apart from the distance-based measures in the following sense. While $D_{F}(P,Q)$ is a symmetric measure ($D_{F}(P,Q) = D_{F}(Q,P)$), it does not satisfy the triangle inequality and is therefore not a metric.

Finally, we note that the variational distance and fidelity are related as follows:
\begin{eqnarray}
 1 - F(P,Q) \leq D_{1}(P,Q) &\leq& \sqrt{1-(F(P,Q))^{2}} \nonumber \\
&=& \sqrt{D_{F}(P,Q)} . \nonumber
\end{eqnarray}

In our discussions below, we also make use of two well-known distance measures between quantum states that are obtained as generalizations of the classical variational distance and fidelity measures~\cite{NCbook}. The trace-distance between quantum states $\rho$ and $\sigma$ is defined as
\[ D_{1}(\rho, \sigma) \equiv \frac{1}{2}\tr\vert \rho -\sigma \vert,\]
where, $|\rho| = \sqrt{\rho^{\dagger}\rho}$ is the positive square-root of $\rho^{\dagger} \rho$. The fidelity of states $\rho$ and $\sigma$ is defined to be
\[ F(\rho, \sigma) \equiv \tr\sqrt{\rho^{1/2}\sigma\rho^{1/2}}. \]
Similar to their classical analogues, the quantum trace-distance and fidelity are also related as
\begin{equation}
1 - F(\rho,\sigma) \leq D_{1}(\rho, \sigma) \leq \sqrt{1-(F(\rho,\sigma))^{2}}. \nonumber
\end{equation}

\subsection{Distance-based incompatibility measures}\label{sec:incompat_meas}

In the first three sections of this paper we will work within the framework of standard quantum theory, and restrict our attention to observables $A$ and $B$ which are self-adjoint operators with purely discrete spectra, and spectral decompositions $A = \sum_{i}a_{i}P^{A}_{i}$ and $B = \sum_{j} b_{j}P^{B}_{j}$. Let ${\rm Pr}^{B}_{\rho} \sim \{p_{\rho}^{B}(j)\}$ denote the probability distribution over the outcomes of a measurement of observable $B$ in state $\rho$. Let ${\rm Pr}^{A\rightarrow B}_{\rho} \sim \{q^{A\rightarrow B}_{\rho}(j)\}$ denote the probability distribution over the outcomes of a $B$ measurement when it follows a measurement of $A$ on the same state $\rho$. These probabilities are given by
\begin{eqnarray}
{\rm Pr}^{B}_{\rho}: p^{B}_{\rho}(j) &=& \tr\left[P^{B}_{j}\rho\right],    \label{eq:prob_defn}\\
 {\rm Pr}^{A\rightarrow B}_{\rho} : q^{A\rightarrow B}_{\rho}(j) &=& \tr\left[P^{B}_{j}\left(\sum_{i}P^{A}_{i}\rho P^{A}_{i}\right)\right]. \nonumber
\end{eqnarray}

If $A$ and $B$ commute, their corresponding eigen-projectors commute, and hence a measurement of $A$ does not affect the distribution of the outcomes of a subsequent measurement of $B$ on the same state. The two probability distributions defined in Eq.~\eqref{eq:prob_defn} above are identical for all states $\rho$ in this case. For a general pair of observables $A$ and $B$, the distance between the probability distributions ${\rm Pr}^{A\rightarrow B}_{\rho}$ and  ${\rm Pr}^{B}_{\rho}$ can thus be regarded as a measure of how much an intervening measurement of $A$ disturbs the statistics of the outcomes of a subsequent measurement of $B$ on the same state $\rho$. Maximizing this distance between probability distributions over all states in the system gives a measure of the mutual incompatibility of the observable $A$ with $B$.

Corresponding to the three distance measures discussed above, we can thus define the following three {\it measures of incompatibility} of observable {\it $A$ with $B$}:
\begin{itemize}
 \item[(i)] $L_{1}$-distance based incompatibility measure:
\begin{equation}
\cQ_{1}(A\rightarrow B) = \sup_{\rho} D_{1}\left({\rm Pr}^{A\rightarrow B}_{\rho},{\rm Pr}^{B}_{\rho} \right). \nonumber
\end{equation}
\item[(ii)] Fidelity-based incompatibility measure:
\begin{equation}
 \cQ_{F}(A\rightarrow B) = \sup_{\rho}\left[1  - F^{2}({\rm Pr}^{A\rightarrow B}_{\rho},{\rm Pr}^{B}_{\rho})\right]. \nonumber
\end{equation}
\item[(iii)] $L_{\infty}$-distance based incompatibility measure:
\begin{equation}
 \cQ_{\infty}(A\rightarrow B) = \sup_{\rho} D_{\infty}({\rm Pr}^{A\rightarrow B}_{\rho},{\rm Pr}^{B}_{\rho}). \nonumber
\end{equation}
\end{itemize}

All three incompatibility measures defined above satisfy,
\begin{equation}
 0 \leq \cQ_{\alpha}(A\rightarrow B) \leq 1, \;\;\; \alpha \in \{1, F, \infty\}
\end{equation}
Furthermore, the lower bound is attained if and only if the observables $A$ and $B$ commute. We state and prove this important property in the following Lemma.
\begin{lemma}
For a pair observables $A$ and $B$ with purely discrete spectra, $Q_{\alpha}(A\rightarrow B) = 0 \; \; (\alpha \in \{1,F,\infty\})$, if and only if $A$ and $B$ commute.
\end{lemma}
\begin{proof}
Recall that the distance $D_{\alpha}\left({\rm Pr}^{A\rightarrow B}_{\rho},{\rm Pr}^{B}_{\rho} \right) = 0$ if and only if the probability distributions ${\rm Pr}^{A\rightarrow B}_{\rho}$ and ${\rm Pr}^{B}_{\rho}$ are identical. It is well known that these two probability distributions coincide for all states $\rho$ if and only if the observables $A$ and $B$ commute~\cite{Davies_book, busch_singh98}. This 
proves the Lemma for $\cQ_{1}$ and $\cQ_{\infty}$.

Similarly, recalling that the fidelity $F\left({\rm Pr}^{A\rightarrow B}_{\rho},{\rm Pr}^{B}_{\rho} \right) = 0$ if and only if the distributions ${\rm Pr}^{A\rightarrow B}_{\rho}$ and ${\rm Pr}^{B}_{\rho}$ are identical, the Lemma is proved for $\cQ_{F}$ as well.
\end{proof}

The incompatibility measures defined here are not symmetric in general. We show in Appendix~\ref{sec:asymmetry}, via an explicit example, that there do exist observables $A, B$, such that,
\[ \cQ_{\alpha}(A\rightarrow B) \neq \cQ_{\alpha}(B \rightarrow A). \]
It is thus natural to define the incompatibility $\cQ_{\alpha}(A,B)$ of the pair of observables $A,B$, as some kind of average of the incompatibilities $\cQ_{\alpha}(A\rightarrow B)$ and $\cQ_{\alpha}(B\rightarrow A)$. This ensures that $\cQ_{\alpha}(A,B)$ is large when {\it both} $\cQ_{\alpha}(A\rightarrow B)$ and $\cQ_{\alpha}(B\rightarrow A)$ are large, and vice-versa.

We therefore propose to define the incompatibility of a set of $N$ observables $\{A_{1}, A_{2}, \ldots, A_{N}\}$ in terms of the pairwise incompatibilities $\{\cQ_{\alpha}(A_{i}\rightarrow A_{j})\}$, in the following manner:
\begin{equation}
\cQ_{\alpha}(A_{1}, A_{2}, \dots, A_{N}) \equiv \frac{1}{N^{2}}\sum_{i,j}\cQ_{\alpha}(A_{i}\rightarrow A_{j}), \label{eq:sym_measureN}
\end{equation}
where $\cQ_{\alpha}(A_{i}\rightarrow A_{i}) = 0$. In particular, the incompatibility of a pair of observables $A$ and $B$, is thus defined as,
\begin{equation}
 \cQ_{\alpha}(A,B) \equiv \frac{\cQ_{\alpha}(A\rightarrow B) + \cQ_{\alpha}(B \rightarrow A)}{4}. \label{eq:sym_measure}
\end{equation}

Incidentally, we may note that an $L_{\infty}$-distance based measure has been used in~\cite{busch_heinosaari08} to characterize approximate joint measurability of general quantum observables. 

Furthermore, there has been renewed interest in the issue of quantifying measurement-induced changes in probabilities, specifically in the context of the so-called Heisenberg error-disturbance relations~\cite{ozawa, diLorenzo, BLW13, BLW14, BLW_jmp} for position and momentum observables. For successive approximate measurements of position and momentum on the same system, these error-disturbance relations seek to provide a trade-off between the ``error", or precision of the position measurement, and ``disturbance", or change in the statistics of a subsequent measurement due to the intervening position measurement. To quantify this disturbance several approaches have been considered, for example, the rms distance between the original momentum operator and the ``disturbed" measurement operator~\cite{ozawa}, or, the difference between the standard deviation of the momentum operator in the original system state and the modified state after an intervening position momentum~\cite{diLorenzo}. A more interesting approach is that of Busch {\it et al.}~\cite{BLW13, BLW14, BLW_jmp}, who use the Wasserstein-2 distance between the probability distribution of the momentum outcomes after the position measurement and the probability distribution of the outcomes of an ideal momentum measurement.

To place our work in the context of these recent discussions, a few remarks are in order. We are seeking to quantify the notion of incompatibility between any pair of observables $A, B$ with purely discrete spectrum, for which a canonically well-defined collapse or change in the state of the system due to measurement exists. We propose to quantify the incompatibility of observable $A$ with another observable $B$ in terms of the change in the statistics of the outcomes of $B$ due to an earlier measurement of $A$. This change in statistics is best measured in terms of the distance between the corresponding probability distributions. Accordingly, in this paper we have considered three well-known measures of distance between probability distributions to give quantitative measures of incompatibility of any pair of observables.

\section{Incompatibility and Disturbance}\label{sec:incompat_disturbance}

With any observable $A$ having a purely discrete spectrum, there is associated a {\it measurement channel} $\cE^{A}$. $\cE^{A}$ is a completely positive trace-preserving (CPTP) map that describes the post-measurement transformation of state $\rho$ after a measurement of $A$, as follows:
\[ \cE^{A}(\rho) = \sum_{i}P^{A}_{i}\rho P^{A}_{i}.\]
Both the trace-distance $D_{1}\left( \cE^{A}(\rho) , \rho\right)$ and the fidelity $F(\cE^{A}(\rho), \rho)$ are valid measures of the {\it disturbance} caused to state $\rho$ by a measurement of $A$~\cite{NCbook}. The {\it maximal disturbance} due to the measurement $A$ can therefore be estimated by either of the following measures:
\begin{eqnarray}
D_{1}^{({\rm max})} (A) &\equiv& \sup_{\rho}\frac{1}{2}\tr \left\vert \cE^{A}(\rho) - \rho \right\vert \nonumber \\
D_{F}^{({\rm max})} (A) &\equiv& 1 - [F^{({\rm min})}(A)]^{2} \nonumber \\
&=& 1 - [\inf_{\rho} F(\cE^{A}(\rho), \rho)]^{2} . \label{eq:disturbance}
\end{eqnarray}
Both these measures of disturbance satisfy,
\begin{equation}
0 \leq D_{1}^{({\rm max})} (A) \leq 1 \; , \; 0 \leq D_{F}^{({\rm max})} (A) \leq 1 \nonumber
\end{equation}

Here we prove an important property of our class of incompatibility measures, namely that the incompatibility of $A$ with $B$ is always upper bounded by the maximal disturbance due to observable $A$.

\begin{lemma}\label{lem:incompat_disturbance}
For a pair of observables $A$ and $B$ with purely discrete spectra, the mutual incompatibility $\cQ_{\alpha}(A\rightarrow B) \;\; (\alpha \in \{1,F,\infty\})$ is bounded above by the maximal disturbance due to the observable $A$. Specifically,
\begin{eqnarray}
 \cQ_{\alpha}(A\rightarrow B) &\leq& D_{1}^{({\rm max)}}(A), \quad \alpha = 1,\infty. \nonumber \\
\cQ_{F}(A\rightarrow B) &\leq& D_{F}^{({\rm max})} (A) =  1  - [F^{({\rm min})} (A)]^{2}
\end{eqnarray}
\end{lemma}

\begin{proof}
We first prove the result for $\alpha \equiv 1, \infty$ and then for $\alpha \equiv F$.

(i) From the definition of $\cQ_{1}(A\rightarrow B)$, we see that,
\begin{eqnarray}
 && \cQ_{1}(A\rightarrow B) = \sup_{\rho}\frac{1}{2}\sum_{j}\vert q^{A\rightarrow B}_{\rho}(j) - p_{\rho}^{B}(j)\vert \nonumber \\
&=& \sup_{\rho}\frac{1}{2}\sum_{j}\left\vert\tr\left[P^{B}_{j}\left(\sum P^{A}_{i}\rho P^{A}_{i}\right)\right] -  \tr[P^{B}_{j}\rho]\right\vert  \nonumber \\
&\leq& \sup_{\rho}\frac{1}{2}\tr\left\vert\sum_{i}P^{A}_{i}\rho P^{A}_{i} - \rho\right\vert \nonumber \\
&=& \sup_{\rho}\frac{1}{2}\tr\left\vert \cE^{A}(\rho) - \rho\right\vert = D_{1}^{({\rm max})}(A) .
\end{eqnarray}
The inequality above follows from the fact that the quantum trace distance between two states is an upper bound on the classical distance between probability distributions obtained by performing measurements on those quantum states. That is,
\begin{equation}
 D_{1}(\rho, \sigma) = \max_{\cM\sim\{M_{i}\}} D_{1}({\rm Pr}^{\cM}_{\rho}, {\rm Pr}^{\cM}_{\sigma}), \label{eq:quant_class}
\end{equation}
where the maximization is over all positive-operator valued measures (POVMs) $\cM\sim\{M_{i}\}$. ${\rm Pr}^{\cM}(\rho) \sim \{\tr[M_{i}\rho]\}$ and ${\rm Pr}^{\cM}(\sigma) \sim \{\tr[M_{i}\sigma]\}$ are the probability distributions arising from the POVM measurement $\cM$ on the states $\rho$ and $\sigma$.

Note that the quantum trace distance also satisfies
\begin{equation}
 D_{1}(\rho, \sigma) = \max_{P}\tr(P(\rho -\sigma)), \label{eq:quant_classD1}
\end{equation}
where the maximization is taken over all projectors $P$. Using this relation we can easily see that,
\begin{eqnarray}
&& \cQ_{\infty}(A\rightarrow B) = \sup_{\rho}\max_{j}\vert q_{j}^{A\rightarrow B}(\rho) - p_{j}^{B}(\rho)\vert \nonumber \\
&=& \sup_{\rho}\max_{j}\left\vert\tr\left[P^{B}_{j}\left(\sum_{i}P^{A}_{i}\rho P^{A}_{i}\right)\right] - \tr[P^{B}_{j}\rho] \right\vert \nonumber \\
&\leq& \sup_{\rho}D_{1}(\sum_{i}P^{A}_{i}\rho P^{A}_{i}, \rho)  = D_{1}^{({\rm max})}(A) .
\end{eqnarray}

(ii) A relation similar to Eq.~\eqref{eq:quant_class} holds for the quantum fidelity between states and the classical fidelity between probability distributions induced by measurements on the states. In particular,
\begin{equation}
 F(\rho, \sigma) = \min_{\cM} F({\rm Pr}^{\cM}_{\rho}, {\rm Pr}^{\cM}_{\sigma}). \label{eq:quant_classF}
\end{equation}
Using this, we can relate the classical fidelity between ${\rm Pr}^{A\rightarrow B}_{\rho}$ and ${\rm Pr}^{B}_{\rho}$ and the quantum fidelity between the states $\rho$ and $\cE^{A}(\rho) = \sum_{i}P^{A}_{i}\rho P^{A}_{i}$, as follows:
\begin{eqnarray}
&& F({\rm Pr}^{A\rightarrow B}_{\rho},{\rm Pr}^{B}_{\rho}) \nonumber \\
&=& \sum_{j}\sqrt{\tr\left[P^{B}_{j}\left(\sum_{i}P^{A}_{i}\rho P^{A}_{i}\right)\right]\tr[P^{B}_{j}\rho]} \nonumber \\
&\geq& F(\sum_{i}P^{A}_{i}\rho P^{A}_{i}, \rho). \nonumber
\end{eqnarray}
This implies,
\begin{eqnarray}
\cQ_{F}(A\rightarrow B) &=& \sup_{\rho}\left[1  - F^{2}({\rm Pr}^{A\rightarrow B}_{\rho},{\rm Pr}^{B}_{\rho})\right] \nonumber \\
&\leq& 1  - \inf_{\rho}F^{2}(\cE^{A}(\rho), \rho) \nonumber \\
&=& 1 - [F^{({\rm min})}(A)]^{2}. \label{eq:max_dist2}
\end{eqnarray}
\end{proof}


\section{Tight upper bounds on incompatibility in finite dimensions} \label{sec:fidelity_bounds}

In this section we focus on the fidelity-based incompatibility measure $\cQ_{F}(A,B)$. We obtain non-trivial upper bounds for observables in a finite-dimensional Hilbert space, and show that the bounds are attained when the observables are totally non-degenerate and correspond to mutually unbiased bases (MUBs).

\begin{theorem}\label{thm:Q2bound}
 For a pair of observables $A$ and $B$ in a $d$-dimensional Hilbert space $\cH_{d}$, the incompatibility of $A$ with $B$ is bounded by
\begin{equation}
  \cQ_{F}(A\rightarrow B) \leq \left(1 - \frac{1}{d} \right). \label{eq:Q2}
\end{equation}
The upper bound in Eq.~\eqref{eq:Q2}is attained when $A$ and $B$ are non-degenerate observables associated with mutually unbiased bases.
\end{theorem}
\begin{proof}
 We have already shown in Section~\ref{sec:incompat_disturbance} that the $\cQ_{F}(A\rightarrow B)$ measure is bounded from above by the maximal disturbance due to observable $A$. It only remains to prove an upper bound on this maximal disturbance. Concavity of the quantum fidelity implies
\begin{eqnarray}
 && \inf_{\rho}F^{2}(\sum_{i}P^{A}_{i}\rho P^{A}_{i}), \rho) \nonumber \\
&\geq& \inf_{|\psi\rangle\langle\psi|} F^{2}(\sum_{i}P^{A}_{i}|\psi\rangle\langle\psi|P^{A}_{i}, |\psi\rangle\langle\psi|) \nonumber \\
&=& \inf_{|\psi\rangle\langle\psi|} \sum_{i} (\langle\psi|P^{A}_{i}|\psi\rangle)^{2} = \inf_{|\psi\rangle\langle\psi|} 2^{ - H_{2}(A ; |\psi\rangle\langle \psi |)} \nonumber \\
&\geq& \frac{1}{d}, \label{eq:fidelity_H2}
\end{eqnarray}
where, in the final step we have used the definition of the second order R{\'e}nyi entropy $H_{2}(A;|\psi\rangle\langle\psi|)$ of the probability distribution arising from a measurement of observable $A$ on state $|\psi\rangle$, and the fact that this $H_{2}$ entropy is always bounded from above by $\log d$. Putting together Eqns.~\eqref{eq:max_dist2}~\eqref{eq:fidelity_H2} we have,
\begin{eqnarray}
 \cQ_{F}(A\rightarrow B) &=& 1 - \inf_{\rho}F^{2}\left({\rm Pr^{A\rightarrow B}_{\rho}}, {\rm Pr^{B}_{\rho}}\right) \nonumber \\
 &\leq& 1  - \inf_{\rho}F^{2}(\cE^{A}(\rho), \rho) \nonumber \\
&\leq& 1 - \inf_{|\psi\rangle\langle\psi|} 2^{ - H_{2}(A | |\psi\rangle\langle \psi |)} \nonumber \\
&\leq& 1 - \frac{1}{d}.
\end{eqnarray}

Finally, to see the bound is tight for mutually unbiased observables, recall that two non-degenerate observables $A$ and $B$ are said to be mutually unbiased iff the corresponding orthonormal eigenbases $A \sim \{|a_{i}\rangle\langle a_{i}|\}$ and $B \sim \{|b_{i}\rangle\langle b_{i}|\}$ satisfy,
\[ |\langle a_{i}| b_{j} \rangle |^{2} = \frac{1}{d}, \; \forall i,j = 1,\ldots, d.\]
Then, a simple calculation shows that the upper bound proved above is attained for an eigenstate of $B$, that is, for $\rho = |\psi\rangle\langle\psi| \equiv |b_{i}\rangle\langle b_{i}|$, for some $i=1,2,\ldots,d.$
\end{proof}

This result immediately gives an upper bound on $\cQ_{F}(A,B)$, the mutual incompatibility of $A$ and $B$.
\begin{eqnarray}
\cQ_{F}(A,B) &=& \frac{\cQ_{F}(A\rightarrow B) + \cQ_{F}(B\rightarrow A)}{4} \nonumber \\
&\leq& \frac{1}{2}\left(1-\frac{1}{d}\right),
\end{eqnarray}
where the upper bound is attained for a pair of non-degenerate observables associated with mutually unbiased bases.

A simple corollary of Theorem~\ref{thm:Q2bound} is a non-trivial upper bound on the average pairwise mutual incompatibility of more than two observables.
\begin{corollary}
The mutual incompatibility of a set of $N$ observables $\{A_{1}, A_{2}, \ldots, A_{N}\}$ in $\cH_{d}$ satisfies,
\begin{equation}
 \cQ_{F}(A_{1}, A_{2}, \ldots, A_{N}) \leq \left(1-\frac{1}{N}\right)\left(1 -\frac{1}{d}\right).
 \end{equation}
 The bound is attained when the observables are non-degenerate and associated with mutually unbiased bases.
\end{corollary}
\noindent{\it Proof.}
The result simply follows from the definition of $\cQ_{F}(A_{1}, A_{2}, \ldots, A_{N})$, and, the upper bound on the incompatibility of each pair of observables in the set $\{A_{1}, A_{2}, \ldots, A_{N}\}$.
\begin{eqnarray}
\cQ_{F}(A_{1}, A_{2}, \ldots, A_{N}) &=& \frac{1}{N^{2}}\sum_{i,j}\cQ_{F}(A_{i}\rightarrow A_{j}) \nonumber \\
&\leq& \frac{N(N-1)}{N^{2}}\left(1-\frac{1}{d}\right) \nonumber \\
&=&  \left(1-\frac{1}{N}\right)\left(1-\frac{1}{d}\right). \; \blacksquare \nonumber
\end{eqnarray}

Note that, both the lower and upper bounds on the mutual incompatibility for a set of $N$ observables obtained here are the same as those obtained for the incompatibility measure $\cQ$ defined in Eq.~\eqref{eq:acc_fid}, as shown in~\cite{incompatibility_BM}. However, as we will see in Sec.~\ref{sec:dc_subspace1} below, the measures $\cQ$ and $\cQ_{F}$ do yield different values for specific pairs of observables.

\subsection{$\cQ_{1}$ and $\cQ_{\infty}$ for a pair of MUBs}

For comparison, we also evaluate the measures $\cQ_{1}(A\rightarrow B)$ and $\cQ_{\infty}(A\rightarrow B)$ for a pair of mutually unbiased observables: $A\sim \{|a_{i}\rangle\langle a_{i}|\}$ and $B\sim\{|b_{j}\rangle\langle b_{j}|\}$. The probability distribution ${\rm Pr_{\rho}^{A\rightarrow B}}: \{q_{\rho}^{A\rightarrow B}(j)\}$, over the outcomes of a $B$ measurement when it follows a measurement of $A$ on the same state $\rho$ is now given by,
\begin{equation}
q_{\rho}^{A\rightarrow B}(j)  = \sum_{i}\langle a_{i}|\rho|a_{i}\rangle |\langle a_{i}|b_{j}\rangle|^{2}  = \frac{1}{d}. \nonumber
\end{equation}
The $\cQ_{1}$ incompatibility measure is therefore given by,
\begin{eqnarray}
 \cQ_{1}(A\rightarrow B) &=& \sup_{\rho}\frac{1}{2} \sum_{j=1}^{d}\vert q_{\rho}^{A\rightarrow B}(j) - p_{\rho}^{B}(j) \vert \nonumber \\
 &=& \sup_{\rho}\frac{1}{2}\sum_{j=1}^{d}\left\vert \frac{1}{d} - \langle b_{j}|\rho|b_{j}\rangle\right\vert \nonumber \\
 &=& 1 - \frac{1}{d}.
\end{eqnarray}
The final step simply follows from the upper bound on the distance between any other $d$-dimensional probability distribution (in this case,  $\{p_{\rho}^{B}(j)\}$) and the uniform distribution. The bound is indeed attained when $\{p_{\rho}^{B}(j) \}$ is a delta distribution, namely,
\[p_{\rho}^{B}(j) = 0 , \; \forall j\neq j_{0} ; \quad p_{\rho}^{B}(j=j_{0})  = 1 , \]
for some $j_{0} \in [1,d]$. The state $\rho$ that induces this distribution is simply a basis state of $B$, that is, $\rho = |b_{j_{0}}\rangle\langle b_{j_{0}}|$, $j_{0} \in [1,d]$. For a pair of mutually unbiased bases the measure $\cQ_{1}$ is indeed symmetric, so that,
\[\cQ_{1}(A,B) = \frac{1}{2}\left(1 - \frac{1}{d}\right).\]

A similar calculation yields,
\[\cQ_{\infty}(A,B) = \frac{1}{2}\left(1 - \frac{1}{d}\right),\]
for a pair of mutually unbiased observables $A$ and $B$ in a $d$-dimensional space.

The above observations lead us to conjecture that both the $\cQ_{1}$ and the $\cQ_{\infty}$ measures are also bounded above by $\frac{1}{2}(1-\frac{1}{d})$, for {\it any} pair of observables in a $d$-dimensional space. 

\subsection{Observables that commute on a subspace} \label{sec:dc_subspace1}

Finally, we present a simple scenario where our approach to quantifying incompatibility goes significantly beyond the standard entropic uncertainty relations formalism. Consider a pair of non-degenerate observables $A, B$ that commute over a subspace of dimension $d_{c}$.  Specifically, we assume that they share $d_{c}$ common eigenvectors, and are mutually unbiased in the $(d-d_{c})$ dimensional subspace where they do not commute.
\begin{eqnarray}
|a_{i}\rangle &=& |b_{i}\rangle, \; \forall \; i=1,\ldots, d_{c} \nonumber \\
|\langle a_{i}| b_{j} \rangle| &=&  \left\{ \begin{array}{ll} 0 &  {\rm for} \; i \leq d_{c}, j > d_{c} \\
0 &  {\rm for} \; i > d_{c}, j \leq d_{c} \\
\frac{1}{\sqrt{d-d_{c}}} & {\rm for} \; i,j > d_{c}
\end{array} \right\} \label{eq:commute_dc}
\end{eqnarray}
For such a pair of observables, our formalism allows us to derive an expression for their mutual incompatibility in terms of the dimension $d_{c}$ of the commuting subspace.
\begin{theorem}\label{thm:QF_dc}
 Consider a pair of non-degenerate observables $A$ and $B$ in $\cH_{d}$ which are such that they have $d_{c}$ common eigenvectors, and their remaining eigenvectors are mutually unbiased (as in Eq.~\eqref{eq:commute_dc}). The mutual incompatibility of such a pair $A$ and $B$ is given by,
\begin{equation}
\cQ_{F}(A, B) = \frac{1}{2}\left(1 - \frac{1}{d-d_{c}}\right). \label{eq:QF_dc}
\end{equation}
\end{theorem}

The proof of this theorem is given in Appendix~\ref{sec:dc_subspace}. In the same section, we also evaluate the mutual incompatibility of the same pair of observables $A$ and $B$ as quantified by the measure $\cQ$ defined in~\cite{incompatibility_BM}. We show that,
\begin{equation}\label{eq:Q_dc}
 \cQ(A,B) \leq \frac{1}{2}\left(1 - \frac{d_{c}+1}{d}\right).
\end{equation}

This example thus highlights clearly the difference between the measures $\cQ$ and $\cQ_{F}$. Comparing Eq.~\eqref{eq:QF_dc} and Eq.~\eqref{eq:Q_dc}, we see that $\cQ_{F}$ and $\cQ$ coincide when $d_{c}=0$ ($A$ and $B$ are mutually unbiased observables) and $d_{c}= (d-1)$ ($A$ and $B$ commute), but do take on different values for $0 < d_{c} < d-1$.

\section{Quantifying Disturbance and Incompatibility for General Measurements}\label{sec:incompat_gen}

In this section, we show how the measures of {\it incompatibility} and {\it disturbance} defined in Sec.~\ref{sec:dist_measure} and Sec.~\ref{sec:incompat_disturbance} can be extended beyond the class of projective measurements. A general observable $\cA$ with discrete outcomes is described by a collection of positive operators $\{ 0 \leq A_{i} \leq \Id\}$ that satisfy $\sum_{i}A_{i} = \Id$. The probability of obtaining outcome $i$ when measuring observable $\cA$ in state $\rho$ is given by $\tr[\rho A_{i}]$. In order to define the class of incompatibility measures $\cQ_{\alpha}$, we also need to specify how the state $\rho$ transforms under a measurement of $\cA$.

In standard quantum theory, there is a canonical association (via the so called {\it Von Neumann - L\"uders} collapse postulate) between an observable characterized by a self-adjoint operator with purely discrete spectrum ($A  = \sum_{i}\alpha_{i}P^{A}_{i}$), and an associated projective measurement ($\cE^{A}(\rho) = \sum_{i}P^{A}_{i}\rho P^{A}_{i}$). For more general observables given by positive operator valued (POV) measures $\cA\sim\{A_{i}\}$, there is no such canonical specification; the associated measurement transformation can now be chosen as any {\it instrument} $\Phi^{\cA}$ {\it implementing} the POV measure $\cA$~\cite{Heinosaari_book}.

An instrument $\Phi^{\cA}$ implementing a measurement of $\cA$ is a collection of completely positive linear maps $\Phi^{\cA}_{i}$ such that, the state $\rho$ transforms to $\Phi_{i}^{\cA}(\rho)$ when outcome $i$ is realized. The probability of realizing outcome $i$ is given by $\tr[\Phi_{i}^{\cA}(\rho)] = \tr[\rho A_{i}]$, for all states $\rho$. The overall transformation of state $\rho$ by instrument $\Phi^{\cA}$ is described by a {\it quantum channel}, that is, a completely positive trace-preserving (CPTP) map (also denoted by $\Phi^{\cA}$):
\[\Phi^{\cA}(\rho) = \sum_{i}\Phi^{\cA}_{i}(\rho).\]

The same observable can indeed be implemented by several different instruments. One simple implementation of a measurement of observable $\cA\sim\{A_{i}\}$ is given by the {\it L\"uders instrument} $\Phi_{\cL}^{\cA}$, in which the post-measurement state after a measurement of observable $\cA$ on state $\rho \in \cH_{d}$ is given by
\[\Phi_{\cL}^{\cA}(\rho) = \sum_{i=1}A_{i}^{1/2}\rho A_{i}^{1/2}.\]
We can now extend our measures of incompatibility and disturbance, for general observables $\cA, \cB$ described by discrete POV measures. However, now the measures indeed crucially depend on the choice of associated instruments $\Phi^{\cA}$ and $\Phi^{\cB}$. The probability distribution over the outcomes of $\cB$ in state $\rho$ is given by,
\[{\rm Pr^{\Phi^{\cB}}_{\rho}} : p^{\Phi^{\cB}}_{\rho}(i) =  \tr[\rho B_{i}] . \]
When the measurement of $\cB$ is preceded by a measurement of $\cA$ on the same state $\rho$, the probability distribution is modified as ${\rm Pr^{\Phi^{\cA}\rightarrow\Phi^{\cB}}}_{\rho}$:
\[{\rm Pr_{\rho}^{\Phi^{\cA}\rightarrow\Phi^{\cB}}} : q_{\rho}^{\Phi^{\cA}\rightarrow\Phi^{\cB}}(i) =  \tr\left[\left(\sum_{j}A^{1/2}_{j}\rho A^{1/2}_{j}\right) B_{i}\right].\]

We can define the incompatibility of $\Phi^{\cA}$ with $\Phi^{\cB}$ as
\begin{equation}
\cQ_{\alpha}(\Phi^{\cA} \rightarrow \Phi^{\cB}) = \sup_{\rho}D_{\alpha}({\rm Pr_{\rho}^{\Phi^{\cA}\rightarrow\Phi^{\cB}}}, {\rm Pr^{\Phi^{\cB}}_{\rho}}).
\end{equation}
As before, we have, $0 \leq \cQ_{\alpha} \leq 1$. Similarly, the maximal disturbance due to a measurement of $\Phi^{\cA}$ can be defined as
\begin{equation}
D_{\alpha}^{\rm max}(\Phi^{\cA}) = \sup_{\rho} D_{\alpha}(\Phi^{\cA}(\rho), \rho),
\end{equation}
where once again, $0\leq D_{\alpha}(\Phi^{\cA}) \leq 1$.

Then, using the relation between the classical and quantum distance measures, it is easy to see that Lemma~\ref{lem:incompat_disturbance} also holds for general quantum measurements. In other words, the incompatibility of $\Phi^{\cA}$ with $\Phi^{\cB}$, $\cQ_{\alpha}(\Phi^{\cA}\rightarrow \Phi^{\cB})$ is always bounded from above by the maximum disturbance $D_{\alpha}^{\rm max}(\Phi^{\cA})$ due to the channel $\Phi^{\cA}$:
\begin{equation}
 \cQ_{\alpha}(\Phi^{\cA}\rightarrow\Phi^{\cB}) \leq D_{\alpha}^{\rm max}(\Phi^{\cA}). \label{eq:gen_incompat_dist}
 \end{equation}

We now show that there exists a non-trivial upper bound on the maximal disturbance due to the L\"uders channel corresponding to a discrete POV measure with a finite number of outcomes. This in turn gives us a non-trivial upper bound on the fidelity-based incompatibility measure $\cQ_{F}(\Phi_{\cL}^{\cA}\rightarrow \Phi_{\cL}^{\cB})$ for such a pair of POV measures $\cA$ and $\cB$.

\begin{theorem}
 The incompatibility of a pair of general observables $\cA$ and $\cB$, with finite number of outcomes $N_{A}$ and $N_{B}$ respectively, and corresponding L\"uders channels:
\[\Phi_{\cL}^{\cA}(\rho) = \sum_{i=1}^{N_{A}}A_{i}^{1/2}\rho A_{i}^{1/2}; \; \Phi_{\cL}^{\cB}(\rho) = \sum_{j=1}^{N_{B}}B_{j}^{1/2}\rho B_{j}^{1/2},\]
is bounded by
\begin{equation}
 \cQ_{F}(\Phi_{\cL}^{\cA}\rightarrow\Phi^{\cB}_{\cL}) \leq 1 - \frac{1}{N_{A}}.
\end{equation}
\end{theorem}
\begin{proof}
The result follows once we prove an upper bound on the maximal disturbance $D_{F}^{\rm max}(\Phi_{\cL}^{\cA})$.  Note that,
\begin{eqnarray}
&& \inf_{|\psi\rangle\langle\psi|}F^{2}\left(\Phi_{\cL}^{\cA}(|\psi\rangle\langle\psi|), |\psi\rangle\langle\psi|\right) \nonumber \\
&=& \inf_{|\psi\rangle\langle\psi|}\sum_{i=1}^{N_{A}}\langle \psi|A^{1/2}_{i}|\psi\rangle^{2} \nonumber \\
&\geq& \inf_{|\psi\rangle\langle\psi|}\sum_{i=1}^{N_{A}}\langle \psi|A_{i}|\psi\rangle^{2} = \inf_{|\psi\rangle\langle\psi|}\sum_{i=1}^{N_{A}} \left[p^{\Phi_{\cL}^{\cA}}_{|\psi\rangle}(i)\right]^{2} \nonumber \\
&=& \inf_{|\psi\rangle\langle\psi|} 2^{-H_{2}\left({\rm Pr^{\Phi_{\cL}^{\cA}}_{|\psi\rangle}}\right)} \geq \frac{1}{N_{A}}.
\end{eqnarray}
Therefore,
\begin{eqnarray}
D_{F}^{\rm max}(\Phi_{\cL}^{\cA}) &=& 1 - \inf_{\rho}F^{2}(\Phi_{\cL}^{\cA}(\rho), \rho) \nonumber \\
&\leq& 1 - \inf_{|\psi\rangle\langle\psi|}F^{2}(\Phi_{\cL}^{\cA}(|\psi\rangle\langle\psi|), |\psi\rangle\langle\psi|) \nonumber \\
&\leq& 1 - \frac{1}{N_{A}}.
\end{eqnarray}
The upper bound on $\cQ_{F}(\Phi_{\cL}^{\cA}\rightarrow\Phi_{\cL}^{\cB})$ now follows from the extension of Lemma~\ref{lem:incompat_disturbance} given in Eq.~\eqref{eq:gen_incompat_dist}.
\end{proof}

Note that, while we have proved a non-trivial (strictly less than one) upper bound on the maximal disturbance due to a L\"uders instrument with a finite number of outcomes, there does not exist such a non-trivial upper-bound for the maximal disturbance due to more general instruments. For example, consider the instrument $\Phi^{\cZ_{p}}$ corresponding to a $Z$-channel, namely,
\[\Phi^{\cZ_{p}}(\rho) = pZ\rho Z+ (1-p)\rho .  \; (0\leq p\leq 1)\]
It is easy to check that the minimal fidelity $\inf_{\rho}F^{2}(\Phi^{\cZ_{p}}(\rho), \rho) = 1- p$, and therefore the maximal disturbance $D_{F}^{\rm max}(\Phi^{\cZ_{p}}) = p \leq 1$.

\section{Conclusions}\label{sec:concl}

To summarize, we have proposed a novel approach to quantify the mutual incompatibility of quantum observables, in terms of the change caused by a measurement of one observable on the statistics of the outcomes of a subsequent measurement of the other observable. We use a class of distance measures between classical probability distributions to quantify this change in statistics, thus leading to a class of incompatibility measures.

In particular, we take a closer look at one such measure based on the classical fidelity between probability distributions. We obtain a tight, non-trivial (strictly less than one) upper bound for the fidelity-based incompatibility of a pair of projective measurements in finite dimensions, and show that this bound is attained for a pair of mutually unbiased bases. Interestingly, the upper bound derived here coincides with that for a different measure of incompatibility, based on the cryptographic notion of accessible fidelity, proposed recently~\cite{incompatibility_BM}. The formalism presented here is however completely general and extends beyond projective measurements. In particular, we use our measure to obtain a non-trivial bound on the mutual incompatibility of a pair of L\"uders instruments with a finite number of outcomes.

Our analysis here brings to light an elegant quantitative connection between operationally motivated notions of {\it disturbance} and {\it incompatibility} for general quantum measurements. Furthermore, since our class of measures vanish if and only if the observables in question commute, this approach goes beyond uncertainty relations in quantifying incompatibility. Interestingly, even optimal entropic uncertainty relations formulated for the successive measurement scenario yield only a trivial (zero) bound for observables that share a single common eigenvector~\cite{MDS03}.

We note that the class of measures presented here is indeed distinct from the incompatibility measure defined in~\cite{incompatibility_BM} based on the accessible fidelity, though both measures coincide for the limiting cases of commuting and mutually unbiased observables. While the operational setting motivating the new class of measures introduced here is a commonly encountered one in the context of quantum cryptography, it remains to be seen if these measures can play a direct role in analyzing the security of quantum cryptographic protocols. Another interesting line of investigation would be to check  whether the measures defined here can be computed efficiently using convex optimization techniques.

\subsection*{ Acknowledgements} The authors are grateful to Somshubhro Bandyopadhyay for valuable discussions and for his comments on an earlier version of the paper.

\appendix

\section{Incompatibility of observables that commute on a subspace}\label{sec:dc_subspace}

Here we prove Theorem~\ref{thm:QF_dc} and obtain an expression for the mutual incompatibility of the pair of observables $A$ and $B$ described in Eq.~\eqref{eq:commute_dc}.

\noindent{\it Proof:}
We first prove an upper bound on $\cQ_{F}(A\rightarrow B)$, the incompatibility of $A$ with $B$. For the pair of observables defined in Eq.~\eqref{eq:commute_dc}, the fidelity between the two relevant probability distributions is given by,
\begin{eqnarray}
&& F({\rm Pr}^{A\rightarrow B}_{|\psi\rangle\langle\psi|}, {\rm Pr}^{B}_{|\psi\rangle\langle\psi|}) \nonumber \\
&=& \sum_{j}\sqrt{\sum_{i}|\langle a_{i}|\psi\rangle |^{2}|\langle a_{i}|b_{j}\rangle |^{2}|\langle b_{j}|\psi\rangle |^{2}} \nonumber \\
&=& \sum_{j \leq d_{c}} |\langle a_{j}|\psi\rangle |^{2} + \sum_{j>d_{c}}|\langle b_{j}|\psi\rangle|\frac{\sqrt{1-\sum_{i\leq d_{c}}|\langle a_{i}|\psi\rangle |^{2}}}{\sqrt{d-d_{c}}} \nonumber
\end{eqnarray}
Since $\sum_{j > d_{c}}|\langle b_{j}|\psi\rangle | \geq \sqrt{\sum_{j> d_{c}}|\langle b_{j}|\psi\rangle|^{2}}$,  and $|a_{j}\rangle \equiv |b_{j}\rangle$ for $j \leq d_{c}$, we obtain,
\begin{eqnarray}
&& F({\rm Pr}^{A\rightarrow B}_{|\psi\rangle\langle\psi|}, {\rm Pr}^{B}_{|\psi\rangle\langle\psi|}) \nonumber \\
&\geq& \sum_{j \leq d_{c}}|\langle a_{j}|\psi\rangle|^{2} + \frac{1 - \sum_{j\leq d_{c}}|\langle b_{j}|\psi\rangle|^{2}}{\sqrt{d-d_{c}}} \nonumber \\
&=& \frac{\sqrt{d-d_{c}}\sum_{j\leq d_{c}}|\langle b_{j}|\psi\rangle|^{2} + (1-\sum_{j\leq d_{c}}|\langle b_{j}|\psi\rangle|^{2})}{\sqrt{d-d_{c}}} \nonumber \\
&\geq& \frac{1}{\sqrt{d-d_{c}}} .
\end{eqnarray}
The bound on $\cQ_{F}(A\rightarrow B)$ follows immediately:
\begin{equation}
\cQ_{F}(A\rightarrow B) \leq 1 - \frac{1}{d-d_{c}} . \nonumber
\end{equation}
This maximal value is attained for eigenstates of $A$ that span the non-commuting subspace, that is, for states $|\psi\rangle = |a_{i}\rangle \;  (i > d_{c})$, and hence,
\[\cQ_{F}(A\rightarrow B) = 1 - \frac{1}{d-d_{c}}. \]

Similarly, we can show, that the incompatibility of $B$ with $A$ is given by,
\begin{equation}
\cQ_{F}(B \rightarrow A) = 1 - \frac{1}{d-d_{c}}. \nonumber
\end{equation}
Together, we get the desired result on the mutual incompatibility of $A$ and $B$ stated in Eq.~\eqref{eq:QF_dc}.
$\blacksquare$

We now consider the mutual incompatibility of the same pair of observables $A$ and $B$ as quantified by the measure $\cQ$ defined in~\cite{incompatibility_BM}. For a set of $N$ non-degenerate observables $\{A^{(1)}, A^{(2)}, \ldots, A^{(N)}\}$ in a $d$-dimensional Hilbert space $\cH_{d}$ with associated eigenvectors $\{|a^{(i)}_{j}\rangle, j=1,\ldots,d.\} $, let $\cE^{A^{(i)}}(\rho) = \sum_{m=1}^{d}(|a^{(i)}_{j}\rangle\langle a^{(i)}_{j}|)\rho (|a^{(i)}_{j}\rangle\langle a^{(i)}_{j}|)$ denote the post-measurement state associated with a measurement of observable $A^{(i)}$ on state $\rho$. Then, it was shown in~\cite{incompatibility_BM} that $\cQ(A^{(1)}, A^{(2)}, \ldots, A^{(N)})$ maybe defined as the complement of the best possible average fidelity an eavesdropper can obtain in a quantum key distribution (QKD) protocol.

The measure $\cQ$ defined in~\cite{incompatibility_BM} can be evaluated as
\begin{eqnarray}
&& \cQ(A^{(1)}, A^{(2)}, \ldots, A^{(N)}) \label{eq:acc_fid} \\
&=& 1 - \sup_{\{|\xi_{m}\rangle\langle\xi_{m}|\}}\frac{1}{Nd}\sum_{m}\lambda_{\rm max}\left[\left(\sum_{i}\cE^{A^{(i)}}(|\xi_{m}\rangle\langle \xi_{m}|)\right)\right], \nonumber
\end{eqnarray}
where, $\lambda_{\rm max}[\Phi]$ denotes the maximum eigenvalue of $\Phi$. The supremum is taken over all positive operator valued measures (POVMs) comprising (non-normalized) rank-one operators $\{|\xi_{m}\rangle\langle\xi_{m}|\}$ such that $\sum_{m}|\xi_{m}\rangle\langle\xi_{m}| = \Id$. 
Thus, $\cQ(A,B)$ is given by,
\begin{widetext}
\begin{eqnarray}
 \cQ(A, B) &=& 1 -   \sup_{\{|\xi_{m}\rangle\langle\xi_{m}|\}}\frac{1}{2d}\sum_{m}\lambda_{\rm max}\left[\cE^{A}(|\xi_{m}\rangle\langle \xi_{m}|) + \cE^{B}(|\xi_{m}\rangle\langle \xi_{m}|)\right] \nonumber \\
 &=& 1  -  \sup_{\{|\xi_{m}\rangle\langle\xi_{m}|\}}\frac{1}{2d}\sum_{m}\lambda_{\rm max}\left[\sum_{i=1}^{d}|\langle a_{i}|\xi_{m}\rangle|^{2}|a_{i}\rangle\langle a_{i}|  + \sum_{j=1}^{d}|\langle b_{j}|\xi_{m}\rangle|^{2}|b_{j}\rangle\langle b_{j}|\right] \nonumber
\end{eqnarray}
\end{widetext}

While the supremum is to be taken over all POVMs $\{|\xi_{m}\rangle\langle\xi_{m}\}$, choosing $|\xi_{m}\rangle = |b_{m}\rangle, \; \forall m =1,\ldots,d$ provides a lower bound on the second term. For the observables $A$ and $B$ defined in Eq.~\eqref{eq:commute_dc}, this lower bound is easily evaluated:
\begin{widetext}
\begin{eqnarray}
\cQ(A,B) &=& 1  - \sup_{\cM \sim \{|\xi_{m}\rangle\langle\xi_{m}|\}}\frac{1}{2d}\sum_{m}\lambda_{\rm max}\left[\sum_{i=1}^{d}|\langle a_{i}|\xi_{m}\rangle|^{2}|a_{i}\rangle\langle a_{i}|  + \sum_{j=1}^{d}|\langle b_{j}|\xi_{m}\rangle|^{2}|b_{j}\rangle\langle b_{j}|\right]  \nonumber \\
&\leq& 1  - \frac{1}{2d}\sum_{m=1}^{d_{c}}\lambda_{\rm max}\left[2|a_{m}\rangle\langle a_{m}|\right] - \frac{1}{2d}\sum_{m=d_{c}+1}^{d}\lambda_{\rm max}\left[\frac{1}{d-d_{c}}\sum_{i>d_{c}}|a_{i}\rangle\langle a_{i}| + |b_{m}\rangle\langle b_{m}|\right] \nonumber \\
&=& 1 - \frac{2d_{c}}{2d} - \frac{1}{2d}\sum_{m=d_{c}+1}^{d}\max_{k > d_{c}}\left[\frac{1}{d-d_{c}}\sum_{i>d_{c}}|\langle b_{k}|a_{i}\rangle |^{2} + |\langle b_{k}|b_{m}\rangle|^{2}\right] \nonumber \\
&\leq& 1 - \frac{d_{c}}{d} - \frac{1}{2d} \sum_{m>d_{c}}\left(\frac{1}{d-d_{c}} + 1\right) = 1 - \frac{d_{c}}{d} - \frac{d-d_{c}+1}{2d} = \frac{1}{2}\left(1 - \frac{d_{c}+1}{d}\right). \label{eq:acc_fid_dc}
\end{eqnarray}
\end{widetext}

\section{Asymmetry of $\cQ_{F}(A\rightarrow B)$}\label{sec:asymmetry}

As we will see below, the asymmetry of the incompatibility measure for a pair of observables manifests clearly when one of the observables has a degenerate spectrum. We first obtain an expression for the maximal disturbance due to a measurement of an observable with a degenerate spectrum.

Consider the observable $C$ with a spectral decomposition
\[ C = \sum_{i=1}^{r}\alpha_{i} P^{C}_{i}, \]
where, each $P^{C}_{i}$ is a projector on to a $d_{i}$-dimensional subspace $(0 < d_{i} < d)$ with $\sum_{i=1}^{r}d_{i} = d$. Further, let $\{|c_{i}\rangle\}$ be an orthonormal basis of eigenvectors of $C$, such that,
\[ P^{C}_{i} = \sum_{k = d_{1} + \ldots + d_{i-1}+1}^{d_{1}+\ldots d_{i-1} + d_{i}} |c_{k}\rangle\langle c_{k}| . \]
For such an observable $C$, we evaluate the disturbance $D_{F}^{\rm max}(C)$ defined in Eq.~\eqref{eq:disturbance}.

First, note that for any state $|\psi\rangle$, the fidelity with the post-measurement state is bounded by
\begin{eqnarray}
 F^{2}[\cE^{C}(|\psi\rangle\langle\psi|), |\psi\rangle\langle\psi|] &=& \sum_{i=1}^{r}\langle \psi|P^{C}_{i}|\psi\rangle^{2} \nonumber \\
 &=& 2^{-H_{2}(C||\psi\rangle)} \geq \frac{1}{r}, \label{eq:disturbance_C}
\end{eqnarray}
where $H_{2}(C||\psi\rangle)$ is the $H_{2}$-entropy of the probability distribution resulting from a measurement of observable $C$ on state $|\psi\rangle$. The minimum value in Eq.~\eqref{eq:disturbance_C} is attained for any state which satisfies $\langle\psi_{\rm opt}|P^{C}_{i}|\psi_{\rm opt}\rangle = \frac{1}{r}$ for all $i=1,\ldots, r$. For example, the state
\[|\psi_{\rm opt}\rangle = \frac{1}{\sqrt{r}} \left[ |c_{1}\rangle + |c_{d_{1}+1}\rangle + \ldots + |c_{d-d_{r}+1}\rangle \right], \]
yields this lower bound. The maximum disturbance due to a measurement of observable $C$ is therefore given by
\begin{equation}
 D_{F}^{\rm max}(C) = 1 - \frac{1}{r} . \nonumber
\end{equation}

We now present an example of a pair of observables $A$ and $B$ such that $\cQ_{F}(A\rightarrow B) \neq \cQ_{F} (B \rightarrow A)$. We choose $A$ to have a totally non-degenerate spectrum in a $d$-dimensional Hilbert space:
\[ A = \sum_{i=1}^{d} a_{i}|a_{i}\rangle\langle a_{i}|,\]
but choose $B$ to be an observable with a degenerate spectrum. Specifically, let $\{|b_{j}\rangle, j=1,\ldots, d\}$ be a basis that is mutually unbiased with respect to $\{|a_{i}\rangle\}$. Choose observable $B$ to have a spectral decomposition
\[ B  = \beta_{1}P^{B}_{1} + \beta_{2}P^{B}_{2}, \]
where, $P^{B}_{1} = \sum_{j=1}^{m}|b_{j}\rangle\langle b_{j}|$, $(0<m < d)$.

As shown above, the maximal disturbance due to a measurement of $B$ is then given by $D_{F}^{\rm max}(B) = \frac{1}{2}$. Hence,
\[ \cQ_{F} (B \rightarrow A) \leq D_{F}^{\rm max} (B) = \frac{1}{2}, \] for all $A$.

To evaluate $Q_{F}(A\rightarrow B)$, note that
\begin{eqnarray}
&& F^{2}[{\rm Pr}^{A\rightarrow B}_{\psi}, {\rm Pr}^{B}_{\psi}] \nonumber \\
&=& \left[ \sum_{k=1}^{2}\left(\sum_{i}|\langle a_{i}|\psi\rangle|^{2}\langle a_{i}|P^{B}_{k}|a_{i}\rangle\right)^{1/2} \langle\psi|P_{k}^{B}|\psi\rangle^{1/2}\right]^{2} . \nonumber
\end{eqnarray}
Choosing $|\psi\rangle = |b_{j}\rangle$ for some $j=1,\ldots,m$, we have, $\langle a_{i}|\psi\rangle = 1/\sqrt{d}$, and,
\[\langle \psi| P^{B}_{2} |\psi\rangle = 0; \qquad \langle a_{i}|P_{1}^{B}|a_{i}\rangle = \frac{m}{d}, \; \forall i=1,\ldots, d.\]
Therefore, for $|\psi\rangle = |b_{j}\rangle$,
\begin{eqnarray}
&& F^{2}[{\rm Pr}^{A\rightarrow B}_{\psi}, {\rm Pr}^{B}_{\psi}] \nonumber \\
&=& \sum_{i}|\langle a_{i}|\psi\rangle |^{2}\langle a_{i}|P_{1}^{B}|a_{i}\rangle = \frac{m}{d}. \nonumber
\end{eqnarray}

Choosing $m < \frac{d}{2}$, we get,
\[\inf_{|\psi\rangle\langle\psi|} F^{2}[{\rm Pr}^{A\rightarrow B}_{\psi}, {\rm Pr}^{B}_{\psi}] < \frac{1}{2}.\]
Hence,
\begin{equation}
 \cQ_{F}(A\rightarrow B) = \sup_{|\psi\rangle\langle\psi|} (1 - F^{2}[{\rm Pr}^{A\rightarrow B}_{\psi}, {\rm Pr}^{B}_{\psi}]) > \frac{1}{2}.
\end{equation}
Thus, we have a pair of observables $A,B$ such that, $\cQ_{F}(B \rightarrow A) \neq \cQ_{F}(A \rightarrow B)$.


\end{document}